\documentclass[journal]{IEEEtran}
\usepackage{ifthen}
\newboolean{draft}
\setboolean{draft}{false} 

\usepackage{schoolbook}
\usepackage{amsmath}                
\usepackage{amssymb}  
\usepackage{psfrag,graphicx}
\usepackage{epsfig}

\newtheorem{thm}{Theorem}
\newtheorem{lem}{Lemma}

\newtheorem{cor}{Corollary}
   
\newtheorem{rem}{Remark}

\newcommand{\E}[1]{\mathrm{E}(#1)}
\def\tr{\mathrm{tr}\,}

\newcommand{\Hardy}[1]{{\mathcal H_#1}}

\newcommand{\Lebesgue}[1]{{\mathcal L_#1}}
\newcommand{\Htwo}{{\Hardy{2}}}

\newcommand{\Ltwo}{{\Lebesgue{2}}}
\newcommand{\Lone}{{\Lebesgue{1}}}

\newcommand{\Rational}{{\mathcal R}}

\newcommand{\RHinf}{{\Rational \Hardy{\infty}}}

\newcommand{\RLone}{{\Rational \Lebesgue{1}}}
\newcommand{\RLinf}{{\Rational \Lebesgue{\infty}}}

\newcommand{\Smirnov}{\mathcal{N}^+}

\newcommand{\re}{\operatorname{Re}}
\newcommand{\im}{\operatorname{Im}}

\newcommand{\eqdef}{\overset{\underset{\mathrm{def}}{}}{=}}

\newcommand{\R}{\mathbb{R}}

\newcommand{\T}{\mathbb{T}}

\newcommand{\norm}[1]{\left\|{#1}\right\|}

\newcommand{\intT}{\int_{-\pi}^{\pi}}

\newcommand{\sys}[4]{
  \left[
    \begin{array}{c|c}
      #1 & #2\\\hline
      #3 & #4
    \end{array}
  \right]
}

\begin{document}
\title{Optimal Linear Control over Channels with Signal-to-Noise Ratio Constraints}

\author{Erik~Johannesson, 
        Anders~Rantzer,~\IEEEmembership{Fellow,~IEEE,} and~Bo~Bernhardsson
\thanks{Submitted to the IEEE Transactions on Automatic Control on March 30th 2012.}%
\thanks{This work was supported by the Swedish Research Council through the Linnaeus Center LCCC; the European Union's Seventh Framework Programme under grant agreement number 224428, project acronym CHAT; and the ELLIIT Strategic Research Center.}
\thanks{The authors are with the Department
of Automatic Control, Lund University, Box 118, 221 00 Lund, Sweden. 
Phone: +46 46 222 87 87. 
Fax: +46 46 13 81 18. 
E-mail: \{erik, rantzer, bob\}@control.lth.se.}%
\thanks{\copyright 2012 IEEE. Personal use of this material is permitted. Permission from IEEE must be obtained for all other uses, in any current or future media, including reprinting/republishing this material for advertising or promotional purposes, creating new collective works, for resale or redistribution to servers or lists, or reuse of any copyrighted component of this work in other works.}}

\maketitle

\begin{abstract}
We consider a networked control system where a linear time-invariant (LTI) plant, subject to a stochastic disturbance, is controlled over a communication channel with colored noise and a signal-to-noise ratio (SNR) constraint. The controller is based on output feedback and consists of an encoder that measures the plant output and transmits over the channel, and a decoder that receives the channel output and issues the control signal. The objective is to stabilize the plant and minimize a quadratic cost function, subject to the SNR constraint. 

It is shown that optimal LTI controllers can be obtained by solving a convex optimization problem in the Youla parameter and performing a spectral factorization. 
The functional to minimize is a sum of two terms: the first is the cost in the classical linear quadratic control problem and the second is a new term that is induced by the channel noise. 

A necessary and sufficient condition on the SNR for stabilization by an LTI controller follows directly from a constraint of the optimization problem. It is shown how the minimization can be approximated by a semidefinite program. The solution is finally illustrated by a numerical example.
\end{abstract}

\begin{IEEEkeywords}
ACGN channel, control over noisy channels, linear-quadratic-Gaussian control, networked control systems, signal-to-noise ratio (SNR)
\end{IEEEkeywords}

\IEEEpeerreviewmaketitle

\section{Introduction}\label{sec:introduction}
\IEEEPARstart{C}{ommunication} limitations are a fundamental characteristic of networked control. The recent trend of decentralized and large-scale systems has therefore driven an interest in research on how control systems are affected by communication phenomena such as random time delays, packet losses, quantization and noise. A popular approach used for research on the fundamental aspects of communication limitations in control systems is to consider a plant that is controlled over a communication channel, as depicted in Fig.~\ref{fig:loop_general}. Due to the lack of a theory that can handle all limitations at once, the channel model is typically chosen to highlight a specific issue.

\begin{figure}[tb]
  \begin{center}
\ifthenelse{\boolean{draft}}
{\includegraphics[width=.45\hsize]{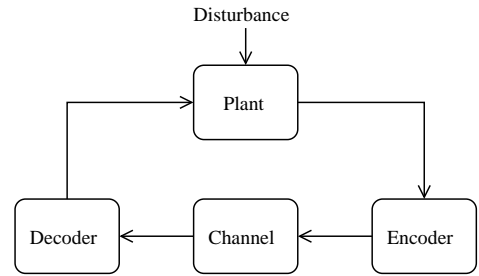}} 
{\includegraphics[width=.7\hsize]{loop_general}}   
    \caption{Feedback control of a plant, with disturbance, over a communication channel. The control system consists of an encoder, which also does measurement filtering, and a decoder, which also determines the control signal.} 
    \label{fig:loop_general}
  \end{center}
\end{figure}

One example is given by data-rate constraints, whose study has led to one of the most well-known results in the area, known as the data-rate theorem. 
In discrete time, it says that an unstable linear plant $G$ can be stabilized over a digital error-free channel if and only if 
\begin{equation}\label{eq:data_rate_theorem}
  \mathcal{R} > \sum_{i} \max \left\lbrace 0,\log_2 \left |\lambda_i(G)\right | \right\rbrace \eqdef \mathcal{H}_G,
\end{equation}
where $\mathcal{R}$ is the rate of the channel and $\lambda_i(G)$ is the ith pole of $G$ \cite{Nair2003, tatikonda04a, Nair2004}.

The situation is more complicated for noisy channels. A necessary and sufficient condition for almost sure asymptotic stabilizability is that the channel capacity $\mathcal{C}$ satisfies $\mathcal{C} > \mathcal{H}_G$ \cite{Matveev2004}.
But since this condition is not generally sufficient for mean-square stability, the concept of any-time capacity has been proposed for characterization of moment stabilizability \cite{sahai06a}. 

Stability is, however, easier to characterize for control of a linear time-invariant (LTI) plant over an additive white Gaussian noise (AWGN) channel or, more generally, an additive colored Gaussian noise (ACGN) channel. Since the communication aspect highlighted by this channel model is a Signal-to-Noise Ratio (SNR) constraint, this setting will be referred to as the SNR framework.

The SNR framework is mainly attractive due to its simplicity. However,
it has been argued that 
the usage of linear controllers admits application of established performance and robustness tools \cite{braslavsky07}. 
The obtained results can sometimes also be used to draw conclusions about and design controllers for other communication limitations such as packet drops \cite{Silva2009} or rate limitations \cite{Silva2010b, silva2010c}. 
Moreover, the SNR framework can be useful for applications such as power control in mobile communication systems \cite{SilvaSNRApproach}. 

\subsection{Main Result and Outline}
This paper considers the problem of control design in the SNR framework. The system has the structure illustrated in Fig. \ref{fig:loop_general}. The plant is LTI, possibly unstable and subject to a Gaussian disturbance. The controller is based on output feedback and consists of an encoder and a decoder. The encoder measures the plant output, filters the measurements and encodes them for transmission over an ACGN channel. The decoder receives the channel output, decodes it and determines the control signal. The objective of the controller is to stabilize the system and minimize a quadratic cost function, while satisfying the SNR constraint. 

The main result is that an optimal LTI output feedback controller can be obtained by minimizing a convex functional and performing a spectral factorization. The minimization is performed over the Youla parametrization of the product of the encoder and the decoder, and the functional is the sum of the classical LQG cost and a new term that is induced by the channel noise. A condition for stabilizability, which coincides with the previously known condition in the AWGN case, is obtained as a constraint of the minimization problem. It is shown how to formulate the minimization as a semidefinite program.
As a by-product of the main result, it is also shown how the encoder and decoder should be chosen in order to minimize the impact of the channel noise while preserving the closed loop transfer function given by a nominal LTI controller that has been designed for a classical feedback system.

The rest of this section will present the previous research in the SNR framework. Section~\ref{sec:notation} presents the mathematical notation used in this paper. The exact problem formulation is given in Section~\ref{sec:problemformulation}. Section~\ref{sec:solution} is devoted to the solution of the problem. Section~\ref{sec:numerics} presents a procedure for numerical solution and a numerical example. Finally, Section~\ref{sec:conclusion} presents the conclusions and discusses further research. Some technical lemmas have been put in the appendix.

\subsection{Previous Research}
Necessary and sufficient conditions for stabilizability, similar to the data-rate theorem, have been found for the SNR framework. They do, however, vary depending on some of the assumptions. Generally, the condition for the AWGN channel is that the SNR $\sigma^2$ satisfies the inequality
\begin{equation}\label{eq:SNR_framework_stabilizability}
  \sigma^2> \left (\prod_i |\max \lbrace 1, \lambda_i(G) \rbrace |^2\right )-1+\eta+\delta,
\end{equation}
where $\eta$ and $\delta$ depend on the specific assumptions, as will be explained.
Assuming no plant disturbance and static state feedback, the condition is that (\ref{eq:SNR_framework_stabilizability}) holds with $\eta=\delta=0$. By writing (\ref{eq:data_rate_theorem}) and (\ref{eq:SNR_framework_stabilizability}) in terms of the respective channel capacities, it can be shown that the capacity requirements for stabilization in the two settings are equal \cite{braslavsky07}.

For LTI output feedback, the condition is again (\ref{eq:SNR_framework_stabilizability}) but now the terms $\eta$ and $\delta$ are non-negative and depend on the non-minimum phase zeros and the relative degree of the plant, respectively \cite{braslavsky07}. If there is a plant disturbance, the same condition holds if  the controller has two degrees of freedom (DOF) \cite{silva2010a} but not if it only has one DOF, in which case the required SNR is larger \cite{Rojas2010}. The condition with $\eta=\delta=0$ can be recovered for the output feedback case, either by introducing channel feedback, meaning that the encoder knows the channel output \cite{silva2010a}, or by allowing a time-varying controller, although the latter leads to poor robustness and sensitivity \cite{braslavsky07}. Further, it has been shown that the condition (\ref{eq:SNR_framework_stabilizability}) with $\eta=\delta=0$ is necessary for stabilizability even if nonlinear and time varying state feedback controllers are allowed \cite{freudenberg10}. 

Similar conditions have been found for LTI control of a plant with no disturbance over an ACGN channel \cite{Rojas2010}. The case with first order moving average channel noise was further analyzed in \cite{middleton2009}.

An early formulation of a feedback control problem over an AWGN channel with feedback
was made in \cite{bansal89}. It was shown how to find the optimal linear controller and that it is globally optimal for first-order plants. A counterexample was provided, showing that non-linear solutions may outperform linear ones for higher order plants \cite{bansal89}. It should, however, be noted that the provided counterexample requires the plant to be time-varying and that the encoder has a memory structure where it does not remember past plant output.

Since then, many authors have considered similar control design problems that have been simplified by assumption of a certain controller structure, see \cite{freudenberg07, goodwin08b, breun2008, li2009, pulgar2011}. A quite general approach was proposed in \cite{silva2010a}, but it was also noted that it leads to a difficult optimization problem with sparsity constraints when it is applied to controllers with two degrees of freedom.

The problem of optimizing the control performance at a given terminal time was considered in \cite{freudenberg08} and \cite{freudenberg09}. The solutions may however yield poor transient performance and can therefore be unsuitable for closed-loop control.

A lower bound on the variance of the plant state was obtained for feedback control over AWGN channels,  using general controllers with two degrees of freedom, in \cite{freudenberg10}. This bound tends to infinity as the SNR approaches the limit for when stabilization is possible. 

An important contribution was recently made in \cite{silva2010c}. Although the paper mainly considers control over a rate-limited channel, this is done through design of an LTI output feedback controller in the SNR framework, assuming an AWGN channel with feedback. The optimal performance is shown to be obtained by solving a convex optimization problem with the same structure as the one obtained in this paper. An optimal controller is then acquired by finding rational transfer functions that approximate certain frequency responses. Related results, for when the controller is pre-designed and the coding system should have unity transfer function, are given in \cite{Silva2010b} and \cite{silva2011}.

Comparing with the results presented here, the case without channel feedback is not mentioned in \cite{silva2010c}. The presented convex functional that gives the optimal cost for the case with channel feedback can, however, be modified to give the optimal cost for this case as well. The expressions for the optimal transfer functions that are given can, with additional work, also be modified to give solutions to the case without channel feedback.

We claim that the solution presented in this paper has a clearer structure than the one given in \cite{silva2010c}. For example, we do not require an over-parametrization of the controller. Moreover, while the plant is assumed to be single-input single-output (SISO) in \cite{silva2010c}, it is here allowed to be slightly more general, making it possible to include any number of noise and reference signals and to penalize the control signal variance. Also, we allow the channel noise to be colored. A final contribution of this paper relative to \cite{silva2010c} is that it is shown how to pose the optimization problem as a semidefinite program.

The approach used in this paper is based on the solution of a communication-theoretic problem involving design of encoders and decoders for a Gaussian source and a Gaussian channel when there is a delay constraint~\cite{joh12IT}. 
Some instances of that problem can be viewed as special cases (open loop versions) of the problem considered here, which therefore may be viewed as a partial generalization.

\section{Notation}\label{sec:notation}
The proofs in this paper make extensive use of concepts from functional analysis, such as $\Lebesgue{p}$ (Lebesgue), $\Hardy{p}$ (Hardy) and $\Smirnov$ (Smirnov) function classes and inner-outer factorizations. To conserve space, only some of the most important facts will be given here. The interested reader is referred either to \cite{joh11phd} or to \cite{Garnett}, \cite{rudin86real} and \cite{Inouye} for the remaining relevant definitions and theorems.

The natural logarithm is denoted $\log$.
The complex unit circle is denoted by $\T$.
For matrix-valued functions $X(z), Y(z)$ defined on $\T$, define 
\begin{equation*}
  \langle X, Y \rangle = \frac{1}{2\pi}\intT \tr\left(X(e^{i\omega})^* Y(e^{i\omega})\right) d\omega
\end{equation*}
and the norms
\begin{align*}
  \norm{X}_1 &= \frac{1}{2\pi} \intT \tr \sqrt{X(e^{i\omega})^*X(e^{i\omega})} \ d\omega \\
  \norm{X}_2 &= \left(\frac{1}{2\pi} \intT \norm{X(e^{i\omega})}_F^2\ d\omega\right)^{1/2} \\
  \norm{X}_\infty &= \text{ess} \sup_{\omega} \sigma_1 \left( X(e^{i\omega})\right),
\end{align*}
where $\norm{\cdot}_F$ is the Frobenius norm and $\sigma_1$ the largest singular value.

A transfer matrix $X(z)$ is said to be proper if the mapping $z\mapsto X(1/z)$ is analytic at $0$. It is strictly proper if also $\lim_{z\rightarrow\infty} X(z) =0$.
The space of all rational and proper transfer matrices with real coefficients is denoted by~$\Rational$.

$\Lebesgue{p}$, for $p=1,2,\infty$, is the space of matrix-valued functions $X(z)$, defined on $\T$, that satisfy $\norm{X}_p < \infty$. The subspaces $\Rational \Lebesgue{p}$ consists of all real, rational and proper transfer matrices with no poles on $\T$. 

$\Hardy{p}$, for $p=1,2,\infty$, is the space of matrix-valued functions $X(z)$ such that $z\mapsto X(1/z)$ is analytic on the open unit disk and
  \[\sup_{r>1} \norm{X_r}_p < \infty,\]
  where $X_r(z) = X(rz)$. 
  The subspaces $\Rational \Hardy{p}$ consists of all real, rational, stable and proper transfer matrices. 
When a function in $\Hardy{p}$ is evaluated on $\T$, it is to be understood as the radial limit $\lim_{r\rightarrow 1^+} X(rz)$. 

The arguments of transfer matrices will often be omitted when they are clear from the context.
Equalities and inequalities involving functions 
evaluated on $\T$ are to be interpreted as holding almost everywhere
on $\T$. 

\section{Problem Formulation}\label{sec:problemformulation}
\begin{figure}[tb]
  \psfrag{C}[]{$C$}
  \psfrag{D}[]{$D$}
  \psfrag{G}[]{$G$}
  \psfrag{v}[]{$v$}
  \psfrag{y}[]{$y$}
  \psfrag{r}[]{$r$}
  \psfrag{t}[]{$t$}
  \psfrag{n}[]{$n$}
  \psfrag{u}[]{$u$}
  \psfrag{w}[]{$w$}
  \psfrag{z}[]{$z$}
  \psfrag{H}[]{$H$}
  \begin{center}
\ifthenelse{\boolean{draft}}
{\includegraphics[width=0.45\hsize]{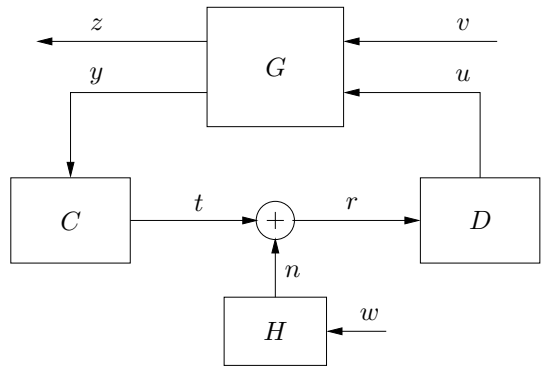}} 
{\includegraphics[width=0.8\hsize]{loop}}   
    
    \caption{Feedback system with ACGN communication channel. The objective is to design the controller components $C$ and $D$ so that the plant $G$ is stabilized and the variance of $z$ is minimized under the SNR constraint $\E{t^2} \leq \sigma^2$. $H$ is a spectral factor of the channel noise $n$.} 
    \label{fig:control_block_diagram}
  \end{center}
\end{figure}

A detailed block diagram representation of the system is shown in Fig.~\ref{fig:control_block_diagram}. The plant $G$ is a multi-input multi-output (MIMO) LTI system with state space realization
\begin{equation*}
  G(z)=\begin{bmatrix} G_{zv}(z) & G_{zu}(z) \\ G_{yv}(z) & G_{yu}(z) \end{bmatrix} =
  \sys{A}
  {
    \begin{array}{cc}B_1 & B_2\end{array}
  }
  {
    \begin{array}{c}C_1 \\ C_2\end{array}
  }
  {
    \begin{array}{cc}D_{11} & D_{12} \\ D_{21} & 0 \end{array}
  },
\end{equation*}
where $(A,B_2)$ is stabilizable and $(C_2,A)$ is detectable. The signals $v$ and $z$ are vector-valued with $n_v$ and $n_z$ elements, respectively. All other signals are scalar-valued. Accordingly, $G_{zv}$ is $n_z \times n_v$, $G_{yv}$ is $1 \times n_v$, $G_{zu}$ is $n_z \times 1$ and $G_{yu}$ is scalar and strictly proper. It is assumed that $G_{zu}^*G_{zu}$ and $G_{yv}G_{yv}^*$ have no zeros or poles on $\T$.

The input $v$ is used to model exogenous signals such as load disturbances, measurement noise and reference signals. It is assumed that $v$ and $w$ are mutually independent white noise sequences with zero mean and identity variance. The other signals in Fig.~\ref{fig:control_block_diagram} are the channel noise $n$, the control signal $u$, the measurement $y$ and the control error or performance index $z$. 

The feedback system is said to be internally stable if no additive injection of a finite-variance stochastic signal at any point in the block diagram leads to another signal having unbounded variance. This is true if and only if all closed loop transfer functions are in $\Htwo$.

The communication channel is an ACGN\footnote{Since only linear controllers are considered, it does not matter if $n$ or $v$ are Gaussian or not. Linear solutions may, of course, be more or less suboptimal depending on their distributions.} channel with SNR $\sigma^2 > 0$. The channel noise has the spectral factor $H\in\RHinf$, which is assumed to have no zeros on~$\T$. 
Since the channel input and output can be scaled by $C$ and $D$, it can be assumed that $n$ has unit variance and thus $\norm{H}_2^2=1$ without loss of generality. The SNR constraint is then equivalent to the power constraint
\begin{equation}\label{eq:channelconstraintintime}
\E{t^2} \leq \sigma^2.
\end{equation}

The objective is to find causal LTI systems $C$ and $D$ that make the system internally stable, satisfy the constraint (\ref{eq:channelconstraintintime}) and minimize 
  $\E{z^T z}$ in stationarity.

By expressing $z$ and $t$ in terms of the transfer functions in Fig.~$\ref{fig:control_block_diagram}$, the objective and the SNR constraint can be written as
\begin{equation*}
  J(C,D) = \norm{G_{zv}+\frac{DCG_{zu}G_{yv}}{1-DCG_{yu}}}_2^2 + \norm{\frac{DHG_{zu}}{1-DCG_{yu}}}_2^2 
\end{equation*}
and
\begin{equation}\label{eq:control_SNR_constraint}
  \norm{\frac{CG_{yv}}{1-DCG_{yu}}}_2^2 + \norm{\frac{DCHG_{yu}}{1-DCG_{yu}}}_2^2 \leq \sigma^2, 
\end{equation}
respectively. The main problem of this paper is thus to minimize $J(C,D)$ over $C$ and $D$ subject to (\ref{eq:control_SNR_constraint}) and internal stability of the feedback system.

For technical reasons, only solutions where the product $DC$ is a rational transfer function will be considered. This may exclude the possibility of achieving the minimum value, but the infimum can still be arbitrarily well approximated by such functions. Since $D$ and $C$ are required to be proper, $DC$ has to be proper as well. That is, $DC\in\Rational$. Though the latter will be enforced, it is not explicitly required that $C$ and $D$ are proper. It will, however, be seen that the solution is constructed so that $C\in\Htwo$ is outer. Then $C,C^{-1}$ are proper, and $D=(DC)C^{-1}$ is also proper.

\section{Optimal Linear Control} \label{sec:solution}
The solution of the problem presented in the previous section is divided into three subsections. The first characterizes internal stability of the system. The second introduces the optimal factorization of a given nominal controller. The third section shows that the optimal factorization result can be used to find an equivalence between the main problem and the minimization of a convex functional in the Youla parameter.

\subsection{Internal Stability}
The product $DC$ will play an important role in the solution. Therefore, introduce 
\begin{equation*}
  K=DC.
\end{equation*}

\begin{figure}[bt]
  \psfrag{C}[]{$C$}
  \psfrag{D}[]{$D$}
  \psfrag{G}[]{$G_{yu}$}
  \psfrag{v}[]{$w_1$}
  \psfrag{y}[]{$y$}
  \psfrag{r}[]{$r$}
  \psfrag{t}[]{$t$}
  \psfrag{n}[]{$n$}
  \psfrag{u}[]{$u$}
  \psfrag{w}[]{$w_2$}
  \begin{center}
\ifthenelse{\boolean{draft}}{\includegraphics[width=0.45\hsize]{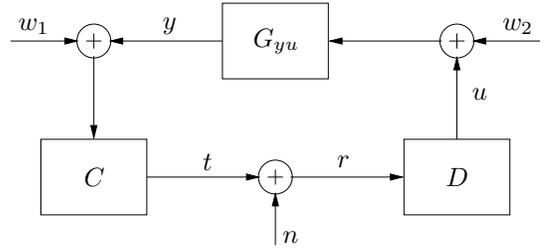}}{\includegraphics[width=0.8\hsize]{loop_stability}}
    \caption{Block diagram for internal stability analysis.} 
    \label{fig:loop_stability}
  \end{center}
\end{figure}
Following the same reasoning as in \cite{zhou96}, it is concluded that internal stability of the systems in Fig.~\ref{fig:control_block_diagram} and Fig.~\ref{fig:loop_stability} are equivalent ($H$ does not have to be included since it is open-loop stable and not part of the feedback loop). The latter can be represented by the closed loop map $T$, defined by
\begin{equation*}
  \begin{bmatrix}
    y \\ t \\ u
  \end{bmatrix} = T
  \begin{bmatrix}
    w_1 \\ w_2 \\ n 
  \end{bmatrix}.
\end{equation*}
Hence, the system in Fig.~\ref{fig:control_block_diagram} is internally stable if and only if
\begin{equation}\label{eq:internal_stability}
  T =
  \begin{bmatrix}
    \dfrac{KG_{yu}}{1-KG_{yu}} & \dfrac{G_{yu}}{1-KG_{yu}} & \dfrac{DG_{yu}}{1-KG_{yu}} \\
    \dfrac{C}{1-KG_{yu}} &\dfrac{CG_{yu}}{1-KG_{yu}} & \dfrac{KG_{yu}}{1-KG_{yu}} \\
    \dfrac{K}{1-KG_{yu}} &\dfrac{KG_{yu}}{1-KG_{yu}} & \dfrac{D}{1-KG_{yu}} 
  \end{bmatrix} \in \Htwo.
\end{equation}

The following two lemmas will give necessary and sufficient conditions for internal stability, respectively.
\begin{lem}\label{lemfnfNecessaryForStability}
  Suppose that $T\in\Htwo$, $G_{yu}=NM^{-1}$ is a coprime factorization over $\RHinf$ and that $U,V \in \RHinf$ satisfy the Bezout identity $VM+UN=1$. Then 
  \begin{align}\label{eqfnfKStabilizing}
    K=\frac{MQ-U}{NQ+V}, \quad Q \in \RHinf.
  \end{align}
\end{lem}
\begin{proof}
  It follows directly from (\ref{eq:internal_stability}) that 
  \begin{gather*}
    \frac{G_{yu}}{1-KG_{yu}} \in \Htwo, \quad \frac{K}{1-KG_{yu}} \in \Htwo, \quad \frac{1}{1-KG_{yu}} \in \Htwo.
  \end{gather*}
  These transfer functions are rational and have no poles on or outside the unit circle, so it follows that 
  \begin{equation}\label{eqfnfClassicalMapStable}
    \begin{bmatrix}
      1 & -K \\ -G_{yu} & 1
    \end{bmatrix}^{-1} =
    \begin{bmatrix}
      \dfrac{1}{1-KG_{yu}} & \dfrac{K}{1-KG_{yu}} \\
      \dfrac{G_{yu}}{1-KG_{yu}} & \dfrac{1}{1-KG_{yu}}
    \end{bmatrix} \in \RHinf,
  \end{equation}
  The set of $K$ satisfying (\ref{eqfnfClassicalMapStable}) can be parametrized using the Youla parametrization of all stabilizing controllers \cite{zhou96}. That is, $K$ can be written as in (\ref{eqfnfKStabilizing}).
\end{proof}

\begin{lem}\label{lemfnfSufficientForStability}
  Suppose that 
  \begin{equation}\label{eqfnfKstab2}
    K=DC = \frac{MQ-U}{NQ+V}, \quad Q \in \RHinf,
  \end{equation}
  where $G_{yu}=NM^{-1}$ is a coprime factorization over $\RHinf$ and $U,V \in \RHinf$ satisfy the Bezout identity $VM+UN=1$. Suppose also that $C\in\Htwo$ is outer and that $D\in\Ltwo$.
  Then $T\in\Htwo$.
\end{lem}
\begin{proof}
  It follows from (\ref{eqfnfKstab2}) that
  \begin{gather*}
    \frac{G_{yu}}{1-KG_{yu}} \in \RHinf, \quad \frac{K}{1-KG_{yu}} \in \RHinf, \\ \frac{KG_{yu}}{1-KG_{yu}} = \frac{1}{1-KG_{yu}} -1 \in \RHinf.
  \end{gather*}
  Moreover,
  \begin{equation*}
    \frac{DG_{yu}}{1-KG_{yu}} =\frac{KG_{yu}}{1-KG_{yu}}C^{-1},
  \end{equation*}
  where the left hand side is in $\Ltwo$ since it is the product of an $\Ltwo$ function and a $\RHinf$ function. Since $C$ is outer, application of Lemma \ref{lem:product_with_inverse_in_Hp} (in the appendix)  gives that the right hand side is in $\Htwo$. A similar argument shows that 
  \begin{equation*}
    \frac{D}{1-KG_{yu}}\in\Htwo.
  \end{equation*}
  Finally,
  \begin{equation*}
    \frac{C}{1-KG_{yu}} \in \Htwo, \quad \frac{CG_{yu}}{1-KG_{yu}} \in \Htwo,
  \end{equation*}
  since these functions are products of an $\Htwo$ function and an $\RHinf$ function.
  Since $\RHinf \subseteq \Htwo$ it has been proved that all elements of $T$ are in $\Htwo$ and so $T\in\Htwo$.
\end{proof}

\subsection{Optimal Factorization}
Suppose for now that $K\in\Rational$ is a given stabilizing controller for the classical feedback system in Fig.~\ref{fig:control_block_diagram_classical}. Thus, $K$ satisfies (\ref{eqfnfKStabilizing}). Nothing else is assumed about the design of $K$. It could for example be the $\Htwo$ optimal controller or have some other desirable properties in terms of step responses or closed loop sensitivity. 

\begin{figure}[tb]
  \psfrag{C}[]{$K$}
  \psfrag{G}[]{$G$}
  \psfrag{v}[]{$v$}
  \psfrag{y}[]{$y$}
  \psfrag{u}[]{$u$}
  \psfrag{w}[]{$w$}
  \psfrag{z}[]{$z$}
  \begin{center}
\ifthenelse{\boolean{draft}}{\includegraphics[width=0.4\hsize]{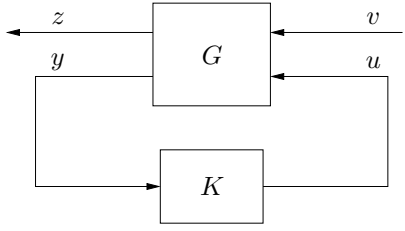}}{\includegraphics[width=0.6\hsize]{loop_classical}}
    \caption{Classical feedback system without communication channel.} 
    \label{fig:control_block_diagram_classical}
  \end{center}
\end{figure}

In either case, it is a natural question to ask what the best way is to implement this controller in the architecture of Fig.~\ref{fig:control_block_diagram}. If the nominal design is to be preserved then $C$ and $D$ should satisfy $K=DC$ since the transfer matrix from $v$ to $z$ would then be the same. Given this relationship, choosing $C$ and $D$ can be thought of as factorizing $K$. The factorization should be chosen to minimize the negative effect of the communication channel. That is, they should keep the system stable, satisfy the SNR constraint and minimize the impact of the channel noise. That is, to minimize the contribution of $n$ to $\E{z^T z}$.

Rewriting $J(C,D)$ and the SNR constraint (\ref{eq:control_SNR_constraint}) with $DC$ replaced by $K$ gives
\begin{equation}\label{eq:y_with_K_fixed}
  J(C,D)=\norm{G_{zv}+\frac{KG_{zu}G_{yv}}{1-KG_{yu}}}_2^2 + \norm{\frac{DHG_{zu}}{1-KG_{yu}}}_2^2
\end{equation}
and
\begin{equation}\label{eq:t_with_K_fixed}
  \norm{\frac{CG_{yv}}{1-KG_{yu}}}_2^2 + \norm{\frac{KHG_{yu}}{1-KG_{yu}}}_2^2 \leq \sigma^2.
\end{equation}

The objective of the optimal factorization problem is to find $C$ and $D$ such that (\ref{eq:y_with_K_fixed}) is minimized subject to (\ref{eq:t_with_K_fixed}) and $K=DC$. Stability is not considered now but it will be seen that the obtained solution is stabilizing anyway. The optimal factorization will later be used to solve the main problem of this paper. Alternatively, it could also be used to factorize a nominal $K$ that was designed for the classical feedback architecture. 

Note that, for given $K$, the first term in (\ref{eq:y_with_K_fixed}) is constant and that the second term is a weighted norm of $D$. In the left hand side of (\ref{eq:t_with_K_fixed}), the first term is a weighted norm of $C$ and the second is constant. Thus, the optimal factorization problem is a minimization of a weighted norm of $D$, subject to an upper bound on a weighted norm of $C$ and the constraint $K=DC$. 

Before the solution to this problem is given, it is noted that the SNR constraint will be impossible to satisfy unless $K$ satisfies
\begin{equation*}
  \alpha \eqdef \sigma^2 - \norm{\frac{KHG_{yu}}{1-KG_{yu}}}_2^2 \geq 0.
\end{equation*}
Actually, if $\alpha=0$ then, since $G_{yv}G_{yv}^*$ has no poles or zeros on $\T$,
\begin{equation*}
  \norm{\frac{CG_{yv}}{1-DCG_{yu}}}_2^2 \!\!\! = 0 \Rightarrow 
  \frac{C}{1-DCG_{yu}} =0 \Rightarrow 
  \frac{KHG_{yu}}{1-KG_{yu}} = 0,
\end{equation*}
which is a contradiction. Thus, it will be assumed that $K$ is such that $\alpha>0$.
Introducing
\begin{equation*}
  S=\frac{1}{1-KG_{yu}}\in\RHinf
\end{equation*}
for notational convenience, the set of feasible pairs $(C,D)$, parametrized by $K$, is defined as
\begin{equation*}
  \Theta_{C,D}(K)=\left \lbrace (C,D):\norm{CSG_{yv}}_2^2 \leq \alpha, DC=K \right \rbrace. 
\end{equation*}
The solution to the optimal factorization problem is now given by the following lemma.

\begin{lem}[Optimal Factorization]\label{lemfnfFactorization}
  Suppose $\alpha > 0$, $S\in \RHinf$, $K\in\Rational$ and that $H\in\RHinf$, $G_{zu}^*G_{zu}\in\RLinf$ and $G_{yv}G_{yv}^* \in \RLinf$ have no zeros on $\mathbb{T}$. Then 
  \begin{equation}\label{eqfnfFactorizationLowerBound}
    \inf_{(C, D)\in\Theta_{C,D}(K)} \norm{DSHG_{zu}}_2^2 \geq \frac{1}{\alpha} \norm{KS^2HG_{zu}G_{yv}}_1^2.
  \end{equation}
  
  Suppose furthermore that $K\in\RLone$ satisfies (\ref{eqfnfKStabilizing}). Then there exists $(C,D)\in\Theta_{C,D}(K)$ with $C\in\Htwo$ outer and \mbox{$D\in\Ltwo$}, such that the minimum is attained and (\ref{eqfnfFactorizationLowerBound}) holds with equality.

  If $K$ is not identically zero, then $(C,D)$ is optimal if and only if $DC = K$ and 
  \begin{equation}\label{eqfnfFactorizationOptimalityCondition}
    \left| C \right|^2 = \frac{\alpha}{\norm{KS^2HG_{zu}G_{yv}}_1}
    \sqrt{\frac{G_{zu}^*G_{zu}}{G_{yv}G_{yv}^*}}  \left|KH\right|
    \hbox{ on } \mathbb{T}.
  \end{equation}
  If $K=0$, then the minimum is achieved by $D=0$ and any $C$ that satisfies $\norm{CSG_{yv}}_2^2 \leq \alpha$.
\end{lem}
\begin{proof}
  If $K=0$ then the right hand side of (\ref{eqfnfFactorizationLowerBound}) is~$0$. Letting $D=0$ gives $\norm{DSHG_{zu}}_2^2=0$ and it is clear that $(C,D)\in\Theta_{C,D}$ if $C$ is as stated. 

  Thus, it can now be assumed that $K$ is not identically zero. Then $C$ is not identically zero and $D=KC^{-1}$.

  By assumption both $G_{zu}^*G_{zu}$ and $G_{yv}G_{yv}^*$ are positive on the unit circle. Since these functions are rational this implies that 
  \begin{equation}\label{eqfnfGzuGyvGreaterThanEpsilon}
    \exists \varepsilon > 0 \text{ such that } G_{zu}^*G_{zu} \geq \varepsilon \text{ and } G_{yv}G_{yv}^* \geq \varepsilon, \text{ on }\T.
  \end{equation}
  Thus by the factorization theorem in \cite{Wiener1959} there exist scalar minimum phase transfer functions $\hat{G}_{zu}, \hat{G}_{yv} \in \Htwo$ such that
  \begin{equation*}
    G_{zu}^*G_{zu} = \hat{G}_{zu}^*\hat{G}_{zu}, \qquad G_{yv}G_{yv}^* = \hat{G}_{yv}\hat{G}_{yv}^*.
  \end{equation*}

  Now, $\norm{CSG_{yv}}_2^2 \leq \alpha$ and Cauchy-Schwarz's inequality gives
  \begin{align*}
    \norm{DSHG_{zu}}_2^2 
    &= \norm{KC^{-1}SH\hat{G}_{zu}}_2^2 \\
    &\geq \frac{\norm{CS\hat{G}_{yv}}_2^2}{\alpha} \norm{KC^{-1}SH\hat{G}_{zu}}_2^2\\
    &\geq \frac{1}{\alpha}\left\langle\left|CS\hat{G}_{yv} \right|,\left|KC^{-1}SH\hat{G}_{zu}\right|\right\rangle^2 \\
    &=\frac{1}{\alpha}\norm{KS^2H\hat{G}_{zu}\hat{G}_{yv}}_1^2 \\
    &=\frac{1}{\alpha} \norm{KS^2HG_{zu}G_{yv}}_1^2. 
  \end{align*}
  This proves the lower bound (\ref{eqfnfFactorizationLowerBound}).
  
  Equality holds if and only if $|KC^{-1}SH\hat{G}_{zu}|$ and $|CS\hat{G}_{yv}|$ are proportional on the unit circle and $\norm{CSG_{yv}}_2^2 = \alpha$. It is easily verified that this is equivalent to (\ref{eqfnfFactorizationOptimalityCondition}). Thus, $(C,D)$ achieves the lower bound if and only if $D=KC^{-1}$ and (\ref{eqfnfFactorizationOptimalityCondition}) holds, since these conditions imply that \mbox{$(C,D)\in\Theta_{C,D}(K)$}.
  
  Under the additional assumptions that $K\in\RLone$ satisfies (\ref{eqfnfKStabilizing}), it will now be shown that there exists such $(C,D)\in\Htwo\times\Ltwo$ with $C$ outer. Since $K$ satisfies~(\ref{eqfnfKStabilizing}) with $M,N,Q,U,V \in \RHinf$ it holds that
  \begin{align*}
    \log{\left|K\right|} = \log{\left|MQ-U\right|} - \log{\left|NQ+V\right|}
  \end{align*}
  By Theorem 17.17 in \cite{rudin86real}, $\log{\left|MQ-U\right|} \in \Lone$ and $\log{\left|NQ+V\right|} \in\Lone$ and thus $\log{\left|K\right|} \in \Lone$. 
  It follows from (\ref{eqfnfGzuGyvGreaterThanEpsilon}) and the boundedness of $H$, $\hat{G}_{yv}$ and $\hat{G}_{zu}$ on $\T$ that
  \begin{equation*}
    \intT \log{\left|\frac{\hat{G}_{zu}}{\hat{G}_{yv}}KH\right|} d\omega > -\infty
  \end{equation*}
  and
  \begin{equation*}
    \left|\frac{\hat{G}_{zu}}{\hat{G}_{yv}}KH\right| \in \Lone.
  \end{equation*}
  Then by the factorization theorem in \cite{Wiener1959} there exists an outer function $C\in\Htwo$ such that (\ref{eqfnfFactorizationOptimalityCondition}) holds. 
  Also, $D=KC^{-1}\in\Ltwo$ since
  \begin{align*}
    \norm{KC^{-1}}_2^2 = \frac{1}{\alpha} \norm{KS^2HG_{zu}G_{yv}}_1 \norm{\frac{K\hat{G}_{yv}}{H\hat{G}_{zu}}}_1<\infty.
  \end{align*}
\end{proof}
\begin{rem}
  The spectral factorization gives some freedom in the choice of $(C,D)$ that attains the bound. For example, $D$ instead of $C$ could be chosen to be $\Htwo$ and outer. That would result in having $C\in\Ltwo$. Considering more solutions than the one selected would require a slightly more complicated stability characterization, so this is not done.
\end{rem}
\begin{rem}
  Optimal $D$ will satisfy
  \begin{equation*}
    \left|D\right|^2 = \frac{\norm{KS^2HG_{zu}G_{yv}}_1}{\alpha} \sqrt{\frac{G_{yv}G_{yv}^*}{G_{zu}^*G_{zu}}} \left|\frac{K}{H}\right|  \hbox{ on } \mathbb{T}.
  \end{equation*}
  It is interesting that the magnitudes of both $C$ and $D$ are directly proportional, on the unit circle, to the square root of the magnitude of $K$. In other words, the dynamics of a nominal controller $K$ is "evenly" distributed on both sides of the communication channel. The static gain of $C$ (and $D$) is tuned so that the SNR constraint is active. In the case when $G_{yv}=G_{zu}$, finding an optimal factorization amounts to performing a spectral factorization of $|KH|$ and tuning the static gain. If also $H=1$ then the magnitudes of the frequency responses of $C$ and $D$ will then be proportional.
\end{rem}

\subsection{Equivalent Convex Problem}
It will now be shown that a solution to the main problem can be obtained, with arbitrary accuracy, by solving a convex minimization problem in the Youla parameter.  

As discussed in the problem formulation, $(C,D)$ should satisfy the SNR constraint (\ref{eq:control_SNR_constraint}) and stabilize the system. The latter corresponds to $T\in\Htwo$ or (\ref{eq:internal_stability}). Also, it was assumed that $DC\in\Rational$.
Thus, the set of feasible $(C,D)$ is given by
\begin{align*}
  \Theta_{C,D}=\left\lbrace(C,D): DC\in\Rational\:, (\ref{eq:control_SNR_constraint}), T\in\Htwo \right\rbrace.
\end{align*}

Let $M,N,U,V$ be determined by a coprime factorization of $G_{yu}$ and introduce
\begin{align}
A&=M^2G_{zu}G_{yv} \label{eq:Adefinition}\\
B&=M^2N^{-1}VG_{zu}G_{yv} \\
E&=MNH \\
F&=(MV-1)H \\
L&=G_{zv}-MN^{-1}G_{zu}G_{yv}. \label{eq:Ldefinition}
\end{align}
It will now be shown that minimization of $J(C,D)$ over $\Theta_{C,D}$ can be performed by minimizing the convex functional
\begin{align*}
  \varphi(Q) &= \norm{L+AQ+B}_2^2 + \frac{\norm{\left(AQ+B\right)\left(EQ+F\right)}_1^2}{\sigma^2-\norm{EQ+F}_2^2},
\end{align*}
over the convex set
\begin{align*}
  \Theta_Q = \left\{ Q : Q\in\RHinf, \norm{EQ+F}_2^2 < \sigma^2 \right\}. 
\end{align*}

The $Q\in\Theta_Q$ obtained from minimizing $\varphi(Q)$ will be used to construct $(C,D)\in\Theta_{C,D}$. However, this will not be possible for $Q$ for which the corresponding $K$ has poles on $\T$. For such $Q$ a small perturbation can then be applied first. This will result in an increased cost, but this increase can be made arbitrarily small. That this is possible is established by the following lemma.
\begin{lem}\label{lemfnfQhat}
  Suppose $Q\in\Theta_Q$ and $\varepsilon>0$. Then there exists $\hat{Q}\in\Theta_Q$ such that 
  \begin{align}\label{eqfnfKFromQhat}
    K=\frac{M\hat{Q}-U}{N\hat{Q}+V} \in \RLone,
  \end{align}
  and
  \begin{equation*}
    \varphi(\hat{Q}) < \varphi(Q)+\varepsilon.
  \end{equation*}
\end{lem}
The proof of Lemma \ref{lemfnfQhat} is based on a perturbation argument and can be found in the Appendix.

The main theorem of the paper can now be stated.
\begin{thm}\label{thmfnfEquivalentProblem}
  Suppose $\sigma^2>0$, that   
  $G_{yu}=NM^{-1}$ is a coprime factorization over $\RHinf$, that $U,V \in \RHinf$ satisfy the Bezout identity $VM+UN=1$ and that $H\in\RHinf$, \mbox{$G_{zu}^*G_{zu}\in\RLinf$} and $G_{yv}G_{yv}^* \in \RLinf$ have no zeros on $\mathbb{T}$. Then
  \begin{equation}\label{eqfnfInfJandInfVarphi}
    \inf_{(C,D)\in\Theta_{C,D}} J(C,D) = \inf_{Q\in\Theta_Q} \varphi(Q).
  \end{equation}
  Furthermore, suppose $Q\in\Theta_{Q}$, $\varepsilon>0$ and let $\hat{Q}\in\Theta_{Q}$ be as in Lemma~\ref{lemfnfQhat}. Then there exists $(C,D)$ such that the following conditions hold:
  \begin{itemize}
  \item If $M\hat{Q}-U $ is not identically zero: $(C,D) \in \Htwo \times \Ltwo$, where $C$ is outer and
    \begin{align}
      K &= \frac{M\hat{Q}-U}{N\hat{Q}+V} \label{eqfnfOptimalK}\\
      \left| C \right|^2 &= \frac{\sigma^2-\norm{\dfrac{KHG_{yu}}{1-KG_{yu}}}_2^2}{\norm{\dfrac{KHG_{zu}G_{yv}}{(1-KG_{yu})^2}}_1}    \sqrt{\frac{G_{zu}^*G_{zu}}{G_{yv}G_{yv}^*}}  \left|KH\right| \hbox{ on } \mathbb{T} \label{eqfnfOptimalC} \\
      D&=KC^{-1} \label{eqfnfOptimalD}
    \end{align}
  \item If $M\hat{Q}-U = 0$: $C=D=0$.
  \end{itemize}
  If $(C,D)$ satisfy these conditions, then $(C,D)\in\Theta_{C,D}$ and
  \begin{equation*}
    J(C,D) < \varphi(Q)+\varepsilon.
  \end{equation*}
\end{thm}

\begin{proof}
  Consider $(C,D) \in \Theta_{C,D}$ and define $K=DC$. Then $(C,D)\in\Theta_{C,D}(K)$ for this choice of $K$. 
  Moreover, because $T\in\Htwo$ it follows from Lemma \ref{lemfnfNecessaryForStability} that $K$ can be written using the Youla parametrization (\ref{eqfnfKStabilizing}). Since the SNR constraint (\ref{eq:control_SNR_constraint}) is satisfied by $(C,D)$ it follows that 
  $K\in\Theta_K$, where $\Theta_K$ is defined by
  \begin{equation*}
    \Theta_{K} = \left \lbrace K: (\ref{eqfnfKStabilizing}), \norm{\frac{KHG_{yu}}{1-KG_{yu}}}_2^2 < \sigma^2 \right \rbrace.
  \end{equation*}
  The inequality in this definition is strict because it was shown earlier that equality cannot hold.
  It has thus been proved that 
  \begin{equation}\label{eqfnfThetaBCDGivesThetaCD}
    (C,D)\in\Theta_{C,D} \Rightarrow (C,D) \in \Theta_{C,D}(K) \text{ for some } K \in \Theta_{K}.
  \end{equation}

  A lower bound will now be determined for $J(C,D)$. This will be accomplished through a series of inequalities and equalities, where each step will be explained afterwards.
  \begin{align*}
    &\inf_{(C,D)\in\Theta_{C,D}} J(C,D) \\
    &\geq \inf_{K \in \Theta_{K}} \inf_{(C,D) \in \Theta_{C,D}(K)} \norm{G_{zv}+\frac{KG_{zu}G_{yv}}{1-KG_{yu}}}_2^2 \! + \norm{\frac{DHG_{zu}}{1-KG_{zu}}}_2^2 \\
    &= \inf_{K \in \Theta_{K}} \! \left[ \norm{G_{zv} \! + \! \frac{KG_{zu}G_{yv}}{1-KG_{yu}}}_2^2 \! \! + \! \! \! \inf_{(C,D) \in \Theta_{C,D}(K)}  \norm{\frac{DHG_{zu}}{1-KG_{zu}}}_2^2 \right] \\
    &\geq \inf_{K \in \Theta_{K}} \norm{G_{zv}+\frac{KG_{zu}G_{yv}}{1-KG_{yu}}}_2^2 + \frac{\norm{\dfrac{KHG_{zu}G_{yv}}{(1-KG_{yu})^2}}_1^2}{\sigma^2-\norm{\dfrac{KHG_{yu}}{1-KG_{yu}}}_2^2} \\
   &= \inf_{Q\in\Theta_{Q}} \varphi(Q)   
  \end{align*}

  The first step follows from (\ref{eqfnfThetaBCDGivesThetaCD}) and rewriting $J(C,D)$ in terms of $K$. In the second step, the first term has been moved out since it is constant in the inner minimization. The third step follows from application of Lemma \ref{lemfnfFactorization} with 
  \begin{equation*}
    \alpha=\sigma^2 - \norm{\frac{KG_{yu}}{1-KG_{yu}}}_2^2>0, \quad S=\frac{1}{1-KG_{yu}}\in\RHinf.
  \end{equation*}
  
Let $A,B,E,F,L$ be given by (\ref{eq:Adefinition})--(\ref{eq:Ldefinition}). Application of the Youla parametrization and the Bezout identity then gives
\begin{multline*}
G_{zv}+\frac{KG_{zu}G_{yv}}{1-KG_{yu}} 
= G_{zv}+\left(\frac{1}{1-KG_{yu}}-1 \right)G_{zu}G_{yv}G_{yu}^{-1} 
\\ = G_{zv} + \left(MNQ+MV-1\right)G_{zu}G_{yv}MN^{-1} 
\\ = AQ+B+G_{zv}-MN^{-1}G_{zu}G_{yv}
\end{multline*}
and
\begin{multline*}
\frac{KHG_{zu}G_{yv}}{(1-KG_{yu})^2} = \frac{G_{zu}G_{yv}G_{yu}^{-1}KHG_{yu}}{(1-KG_{yu})^2} 
\\ = G_{zu}G_{yv}M^2\left( Q+N^{-1}V \right) \left( MNQ+MV-1 \right) H
\\ = (AQ+B)(EQ+F).
\end{multline*}
The fourth step now follows from the definition of $\varphi(Q)$.

  Now a suboptimal solution will be constructed. Suppose that $Q\in\Theta_{Q}$ and $\varepsilon>0$ and let $\hat{Q}\in\Theta_{Q}$ be as given by Lemma \ref{lemfnfQhat} and define $K\in\RLone$ by (\ref{eqfnfOptimalK}). Then $K\in\Theta_{K}$ and
  \begin{equation*}
    \varphi(\hat{Q}) =
    \norm{G_{zv}+\frac{KG_{zu}G_{yv}}{1-KG_{yu}}}_2^2 + \frac{\norm{\dfrac{KHG_{zu}G_{yv}}{(1-KG_{yu})^2}}_1^2}{\sigma^2-\norm{\dfrac{KHG_{yu}}{1-KG_{yu}}}_2^2}
  \end{equation*}

  If $M\hat{Q}-U=0$ then $K=0$, 
  \begin{equation*}
    J(0,0)=\norm{G_{zv}}_2^2=\varphi(\hat{Q}) < \varphi(Q) + \varepsilon,
  \end{equation*}
  and we are done.

  If, on the other hand, $M\hat{Q}-U$ is not identically zero then $K$ is not identically zero.
  By Lemma~\ref{lemfnfFactorization} there then exists an outer $C\in\Htwo$ and $D\in\Ltwo$ such that (\ref{eqfnfOptimalC}) and (\ref{eqfnfOptimalD}) are satisfied. The lemma also says that such $(C,D)$ satisfy 
  \begin{equation*}
    \norm{\frac{DHG_{zu}}{1-KG_{yu}}}_2^2 = \frac{\norm{\dfrac{KHG_{zu}G_{yv}}{(1-KG_{yu})^2}}_1^2}{\sigma^2-\norm{\dfrac{KHG_{yu}}{1-KG_{yu}}}_2^2}
  \end{equation*}
  and
  \begin{equation*}
    \norm{\frac{CG_{yv}}{1-KG_{yu}}}_2^2 \leq \sigma^2-\norm{\dfrac{KHG_{yu}}{1-KG_{yu}}}_2^2.
  \end{equation*}

  $D,C$ and $K$ satisfy the conditions of Lemma \ref{lemfnfSufficientForStability}, so $T\in\Htwo$, which implies that $(C,D)\in\Theta_{C,D}$. Moreover,
  \begin{align*}
    J(C,D) &= \norm{G_{zv}+\frac{KG_{zu}G_{yv}}{1-KG_{yu}}}_2^2 + \norm{\frac{DHG_{zu}}{1-KG_{yu}}}_2^2 \\ &= \varphi(\hat{Q}) = \varphi(Q) + \varepsilon.
  \end{align*}
  Since $\varepsilon$ can be made arbitrarily small this shows that (\ref{eqfnfInfJandInfVarphi}) holds and hence the proof is complete.
\end{proof}

\begin{rem}
  Theorem \ref{thmfnfEquivalentProblem} shows that an $\varepsilon$-suboptimal solution to the main problem  can be found by minimizing $\varphi(Q)$ over $\Theta_{Q}$. The obtained $Q$ may have to be perturbed so that the resulting $K$ has no poles on the unit circle. Then $C$ is given by a spectral factorization and $D$ is then obtained from $C$.
\end{rem}

A by-product of Theorem \ref{thmfnfEquivalentProblem} is a necessary and sufficient criterion for the existence of a stabilizing LTI controller that satisfies the SNR constraint.
\begin{cor}\label{corfnfStabilizationCriteria}
  There exists $(C,D)$ that stabilize the closed loop system of Fig.~\ref{fig:control_block_diagram} subject to the SNR constraint (\ref{eq:control_SNR_constraint}) if and only if there exists $Q\in\RHinf$ such that 
\begin{equation}\label{eqfnfDerivedStabilizabilityCriteria}
\norm{\left (MNQ+MV-1\right )H}_2^2 < \sigma^2.
\end{equation}
For the AWGN channel, we have that $H=1$ and the condition can be written
  \begin{equation}\label{eqfnfDerivedStabilizabilityCriteriaAWGN}
    \norm{MNQ+MV}_2^2 < \sigma^2 +1
  \end{equation}
  since $MNQ+MV-1$ is strictly proper and thus orthogonal to $1$.
\end{cor}
\begin{rem}
  Corollary \ref{corfnfStabilizationCriteria} implies that the minimum SNR compatible with stabilization of a stochastically disturbed plant by an output feedback LTI controller with two degrees of freedom over an ACGN channel can be found by minimizing the left hand side of (\ref{eqfnfDerivedStabilizabilityCriteria}) over $Q\in\RHinf$.  
  
  For the AWGN case, the analytical condition (\ref{eq:SNR_framework_stabilizability}), presented in \cite{braslavsky07}, is actually derived from a minimization of the left hand side of (\ref{eqfnfDerivedStabilizabilityCriteriaAWGN}). This means that the same condition is necessary and sufficient in the present problem setting as well, when the channel noise is white. This fact has been noted before in \cite{silva2010a}. To elaborate, there is no plant disturbance in the setup of \cite{braslavsky07}. In that case, the SNR required for stabilizability will be the same regardless if the controller has one or two degrees of freedom. However, \cite{Rojas2009} considered the case when there is a plant disturbance and showed that the SNR required for stabilizability may then be larger than prescribed by~(\ref{eq:SNR_framework_stabilizability}). 
However, the controller in \cite{Rojas2009} was assumed to only have one DOF (the encoder part was fixed to be a unity gain). Theorem 17 in \cite{silva2010a} and this corollary shows that if the controller has two DOF, then (\ref{eq:SNR_framework_stabilizability}) is again a necessary and sufficient criterion for stabilizability.

For the ACGN case, however, this result is not identical to those in \cite{Rojas2010} and \cite{middleton2009} since they assumed no plant disturbance and (effectively) controllers with one DOF.
\end{rem}

It will now be shown that the minimization of $\varphi(Q)$ over $\Theta_{Q}$ is actually a convex problem. To this end, define the functional
\begin{align*}
  \rho(a,e) &= \frac{1}{2\pi}\intT a(\omega)^2 d\omega + \frac{\left(\frac{1}{2\pi}\intT a(\omega)e(\omega) d\omega\right)^2}{\sigma^2+1-\frac{1}{2\pi}\intT e(\omega)^2 d\omega}
\end{align*}	
with domain $\Theta_\rho$ consisting of functions $a$ and $e$ that are continuous on $[-\pi,\pi]$ and satisfy
\[\frac{1}{2\pi}\intT e(\omega)^2 \ d\omega < \sigma^2+1.\]
\begin{lem}\label{lem:coding_feedb_relaxation_convex}
  The functional $\rho(a,e)$ is convex.
\end{lem}
\begin{proof} Take $n \geq 2$. 
  The function 
  \begin{align*}
    f(x,y,v) &= (x+yv)^T(x+yv)-v^2, \\
    &= x^Tx + 2vx^Ty + v^2(y^Ty-1)
  \end{align*}
  with domain $\left \{ (x,y,v) : x,y \in \mathbb{R}^n,\ v \in \mathbb{R},\ y^Ty < 1\right \}$, is convex in $\left(x,y \right)$ for any $v\in\mathbb{R}$. Thus,
  \begin{align*}
    g(x,y) = \max_{v\in\mathbb{R}} f(x,y,v) = x^Tx + \frac{\left (x^Ty\right )^2}{1-y^Ty},
  \end{align*}
  with domain $\left \{ (x,y) : x,y \in \mathbb{R}^n,\ y^Ty < 1\right \}$, is convex in $(x,y)$ since it is the pointwise maximum of a set of convex functions \cite{Boyd2004}.
  Now, suppose $(a,e)\in\Theta_\rho$. Let
  \begin{align*}
    \omega_1 &= 0, \quad \omega_{k+1}-\omega_k = 2\pi/n, \quad k=1,\ldots, n-1 
    \\
    \hat{a}&=\begin{bmatrix}a(\omega_1) & a(\omega_2) & \ldots & a(\omega_n)	\end{bmatrix}^T \\
    \hat{e}&=\begin{bmatrix}e(\omega_1) & e(\omega_2) & \ldots & e(\omega_n)	\end{bmatrix}^T.   \end{align*}
  By definition of the integral, it holds that
  \begin{align*}
    \lim_{n \rightarrow \infty} \frac{\hat{e}^T\hat{e}}{(\sigma^2+1) n} = \frac{1}{(\sigma^2+1)} \frac{1}{2\pi}\intT e(\omega)^2 d\omega < 1.
  \end{align*}
  So for large $n$, $\left (\hat{a},(\sigma^2+1)^{-1/2}\hat{e} \right )/\sqrt{n}$ belongs to the domain of $g$ and
  \begin{align*}
        \rho(a,e) &= \lim_{n\rightarrow\infty} g\left (\frac{\hat{a}}{\sqrt{n}},\frac{\hat{e}}{\sqrt{(\sigma^2+1) n}}\right ).
  \end{align*}
  Since the right hand side is convex in $(\hat{a},\hat{e})$, and thus in $(a,e)$, it follows that $\rho(a,e)$ is convex.
\end{proof}
\begin{rem}
  Convexity of $\rho(a,e)$ has been shown previously in \cite{derpich2011}. This proof is, however, substantially shorter.
\end{rem}

The convex functional $\rho$ will be used in a relaxation of the minimization of $\varphi(Q)$. A slight modification has to be done to $\varphi(Q)$ in order to be able to compare the two functionals. For this purpose, define
\begin{equation*}
\Delta(Q) = \norm{L}_2^2+2\re\left\langle L,AQ+B\right\rangle.
\end{equation*}
Since $\Delta(Q)$ is affine in $Q$ it doesn't affect the convexity of $\varphi(Q)$. Define the functional
\begin{align*}
  \varphi_0(Q) &= \varphi(Q)-\Delta(Q) \\
  &= \norm{AQ+B}_2^2 + \frac{\norm{\left(AQ+B\right)\left(EQ+F\right)}_1^2}{\sigma^2-\norm{EQ+F}_2^2}.
\end{align*}

\begin{lem}\label{lem:control_relaxation}
  Suppose $Q\in\Theta_{Q}$. Then $\varphi_0(Q)\leq \gamma$ if and only if there exists $(a,e)\in\Theta_\rho$ such that $\rho(a,e) \leq \gamma$ and
  \begin{align}
    a(\omega) &\geq \sqrt{G_{zu}^*G_{zu}G_{yv}G_{yv}^*}\left| M^2Q+\frac{M^2V}{N} \right |, \quad \omega\in [-\pi,\pi] \label{eq:a_larger_than_absolute}\\
    e(\omega) &\geq |EQ+F|, \quad \omega \in [-\pi,\pi].		\label{eq:e_larger_than_absolute}
  \end{align}
\end{lem}
\begin{proof}
Suppose first that $\varphi_0(Q)\leq \gamma$. Let
  \begin{align*}
    a(\omega)&=\sqrt{G_{zu}^*G_{zu}G_{yv}G_{yv}^*}\left| M^2Q+\frac{M^2V}{N} \right | \\
     e(\omega)&=|EQ+F|
  \end{align*}
  and it follows that $(a,e)\in\Theta_\rho$ and $\rho(a,e)=\varphi_0(B,K)$. Conversely, suppose that $(a,e)\in\Theta_\rho$ satisfy (\ref{eq:a_larger_than_absolute}) and (\ref{eq:e_larger_than_absolute}) and that $\rho(a,e) \leq \gamma$. Then it follows from inspection of $\varphi_0(Q)$ and $\rho(a,e)$ that $\varphi_0(Q) \leq \rho(a,e) \leq \gamma$.
\end{proof}

Convexity can now be proved.
\begin{thm}\label{thm:approximation_convex}
  The problem of minimizing $\varphi(Q)$ over $\Theta_{Q}$ is convex.
\end{thm}
\begin{proof}
  Suppose $Q_1,Q_2\in\Theta_{Q}$. Then by Lemma~\ref{lem:control_relaxation} there exists $(a_1,e_1)\in\Theta_\rho$ and $(a_2,e_2)\in\Theta_\rho$ such that $\rho(a_1,e_1) \leq \varphi_0(Q_1)$ and $\rho(a_2,e_2) \leq \varphi_0(Q_2)$. It thus holds for $0\leq\theta\leq 1$ that
  \begin{multline*}
    \theta \varphi_0(Q_1) + (1-\theta) \varphi_0(Q_2)
    \geq\theta\rho(a_1,e_1) + (1-\theta)\rho(a_2,e_2) \\
    \geq \rho\left(\theta a_1 + (1-\theta) a_2,\theta e_1+(1-\theta)e_2\right) \\
     \geq \varphi_0(\theta Q_1+(1-\theta) Q_2).
  \end{multline*}	
  The second inequality follows from Lemma~\ref{lem:coding_feedb_relaxation_convex}. The third inequality follows from Lemma \ref{lem:control_relaxation} and that the constraints (\ref{eq:a_larger_than_absolute}) and (\ref{eq:e_larger_than_absolute}) are convex. 
  It is thus proved that $\varphi_0(Q)$ is convex in $Q$. Then $\varphi(Q)$ is convex since $\Delta(Q)$ is convex. It is finally noted that $\Theta_{Q}$ is a convex set.
\end{proof}

\section{Numerical Solution}\label{sec:numerics}
By Lemma \ref{lem:control_relaxation}, minimizing $\varphi(Q)$ over $\Theta_Q$ is equivalent to minimizing $\rho(a,e)+\Delta(Q)$ over $\Theta_\rho \times \Theta_Q$ subject to (\ref{eq:a_larger_than_absolute}) and (\ref{eq:e_larger_than_absolute}). This problem is infinite-dimensional, so the integrals are discretized for numerical solution. It will now be shown how the discretized problem can be posed as a semidefinite program. 

Let $n \geq 2$ and define $\left\lbrace\omega_{k}\right \rbrace_{k=0}^{n-1}$, $\hat{a}$ and $\hat{e}$ as in the proof of Lemma \ref{lem:coding_feedb_relaxation_convex}.
Approximations of $\rho(a,e)$ and $\Delta(Q)$ with $n$ grid points are then given by
\begin{align*}
  \rho_n(\hat{a},\hat{e}) &= \frac{1}{n}\hat{a}^T\hat{a} + \frac{\left(\frac{1}{n} \hat{a}^T\hat{e}\right)^2}{\sigma^2+1-\frac{1}{n}\hat{e}^T\hat{e}} 
  \\
  \Delta_n(Q) &= \norm{L}_2^2 + \frac{2}{n} \re \sum_{k=1}^n \tr \! \left .\left( L^* (A Q + B) \right)\right |_{z=e^{i\omega_k}}. 
\end{align*}
The accuracy of this approximation clearly depends on the number of grid points $n$. 
When implementing the minimization program, $Q$ is parametrized using a finite basis representation. The accuracy of the approximation obviously depends on this representation as well.

The denominator of $\rho_n(\hat{a},\hat{e})$ is positive for sufficiently large $n$ and $\rho_n(\hat{a},\hat{e}) + \Delta_n(Q)$ can be written as a Schur complement. It follows that \hbox{$\rho_n(\hat{a},\hat{e}) + \Delta_n(Q) \leq \gamma$} if and only if 
\begin{equation*}
  \begin{bmatrix}
    \frac{1}{n}\hat{e}^T\hat{e}-\sigma^2-1 & \frac{1}{n}\hat{a}^T\hat{e} \\
    \frac{1}{n}\hat{e}^T\hat{a} &  \frac{1}{n}\hat{a}^T\hat{a} + \Delta_n(Q) - \gamma
  \end{bmatrix} \preceq 0,
\end{equation*}
or, equivalently,
\begin{equation*}
  \begin{bmatrix}
    n(\sigma^2+1) & 0 \\
    0 & n\gamma-n\Delta_n(Q)
  \end{bmatrix} -
  \begin{bmatrix}
    \hat{e} & \hat{a}
  \end{bmatrix}^T I
  \begin{bmatrix}
    \hat{e} & \hat{a}
  \end{bmatrix} \succeq 0.
\end{equation*}
Noting that the left hand side of the last inequality is also a Schur complement, it follows that it is equivalent to
\begin{equation}\label{eq:control_LMI}
  \begin{bmatrix}
    I & \hat{e} & \hat{a} \\
    \hat{e}^T & n(\sigma^2+1) & 0 \\
    \hat{a}^T & 0 & n\gamma-n\Delta_n(Q)
  \end{bmatrix} \succeq 0.
\end{equation}

Let $g_k = \sqrt{G_{zu}(e^{i\omega_k})^*G_{zu}(e^{i\omega_k})G_{yv}(e^{i\omega_k})G_{yv}(e^{i\omega_k})^*}$.
The constraints can then be approximated by
\begin{align}
  a(\omega_k) &\geq g_k \left|M(e^{i\omega_k})^2 \right | \left |Q(e^{i\omega_k})+\frac{V(e^{i\omega_k})}{N(e^{i\omega_k})} \right|,\ k=1\ldots n \\
  e(\omega_k) &\geq \left|E(e^{i\omega_k})Q(e^{i\omega_k})+F(e^{i\omega_k})\right|,\quad k=1\ldots n \\
  \sigma^2+1 &> \frac{1}{n}\sum_{k=1}^n{e(\omega_k)^2}. \label{eq:control_lastconstraint}
\end{align}
Minimizing $\gamma$ subject to (\ref{eq:control_LMI})--(\ref{eq:control_lastconstraint}) is a semidefinite program.

A procedure for numerical solution will now be outlined.
\begin{enumerate}
\item Determine $N,M,U,V\in\RHinf$ by a coprime factorization of $G_{yu}$ and calculate $A,B,E,F$ and $L$.
\item Choose $n$ large, determine the grid points $\omega_k$, $k=1\ldots n$ and solve the optimization problem of minimizing $\gamma$ subject to (\ref{eq:control_LMI})--(\ref{eq:control_lastconstraint}). The transfer function $Q$ is parametrized with a finite basis representation, for example as an FIR filter. 
  If the problem is infeasible it could mean that a larger $\sigma^2$ is needed
  to stabilize the plant. This can be checked analytically using the
  condition in \cite{braslavsky07}. 
  If $\sigma^2$ is sufficiently large according to this condition, the problem could still become
  infeasible if $n$ is too small or $Q$ is too coarsely parametrized.
\item If $NQ+V$ has zeros on the unit circle, determine a small perturbation $\hat{Q}$ of $Q$ as outlined by Lemma~\ref{lemfnfQhat}.
\item Determine $K$ from (\ref{eqfnfKFromQhat}).
\item	Use a finite basis approximation $A(\omega)$ of $CC^*$, for example the para\-metrization 
  \begin{equation}\label{eq:coding_spectrum_parametrization}
    A(\omega) =A_0 + \sum_{k=1}^{N_c} A_k \left(e^{ki\omega} + e^{-ki\omega}\right),
  \end{equation}
and fit $A(\omega)$ to the right hand side of (\ref{eqfnfOptimalC}), for example by mini\-mizing the mean squared deviation.
\item Perform a spectral factorization of $A(\omega)$, choosing $C$ as the stable and minimum phase spectral factor.
\item Let $D=KC^{-1}$.
\end{enumerate}

\subsection{Example}

\begin{figure}[b]
  \psfrag{C}[]{$C$}
  \psfrag{D}[]{$D$}
  \psfrag{G}[]{$P$}
  \psfrag{v}[]{$w$}
  \psfrag{y}[]{$y$}
  \psfrag{r}[]{$r$}
  \psfrag{t}[]{$t$}
  \psfrag{n}[]{$n$}
  \psfrag{u}[]{$u$}
  \begin{center}
\ifthenelse{\boolean{draft}}{\includegraphics[width=0.45\hsize]{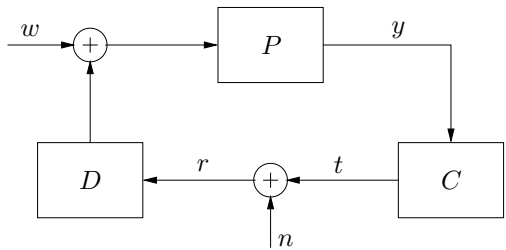}}{\includegraphics[width=0.75\hsize]{loop_SISO}}
    \caption{Control of a SISO plant over an AWGN channel.} 
    \label{fig:loop_SISO}
  \end{center}
\end{figure}

Consider the system in Fig.~\ref{fig:loop_SISO}. A SISO plant is controlled over an AWGN channel. The SISO plant represents a special case where 
\begin{align*}
  G(z)=\begin{bmatrix} G_{zv}(z) & G_{zu}(z) \\ G_{yv}(z) & G_{yu}(z) \end{bmatrix} =
  \begin{bmatrix}
    P(z) & P(z) \\P(z) & P(z)
  \end{bmatrix}.
\end{align*}

Let the plant be $P(z)=1/(z(z-2))$. It has one unstable pole and a one-sample time delay. Using the stabilizability condition (\ref{eq:SNR_framework_stabilizability}), it is determined that stabilization is possible for $\sigma^2 > 12$. (We have $\eta=0$, since there are no non-minimum phase zeros, and $\delta=9$, because of the location of the unstable pole and the relative degree, which is $2$. For details, see \cite{braslavsky07}). 

  A controller was determined for various values of $\sigma^2$, using the algorithm outlined above. The optimization was performed in Matlab, using the toolboxes Yalmip \cite{yalmip} and SeDuMi \cite{sedumi}. In the optimization program, $n=629$ grid points were used and $Q$ was parametrized as an FIR filter with length~$20$.
  The plant output variance is plotted in Fig.~\ref{fig:performance}
  for a number of different~$\sigma^2$. It can be seen that the variance grows unbounded as $\sigma^2$ approaches $12$ and the feedback system comes closer to instability. This seems to be in agreement with the performance bound given in \cite{freudenberg10}.
  \begin{figure}[tb]
    \begin{center}
\ifthenelse{\boolean{draft}}{\includegraphics[width=0.5\hsize]{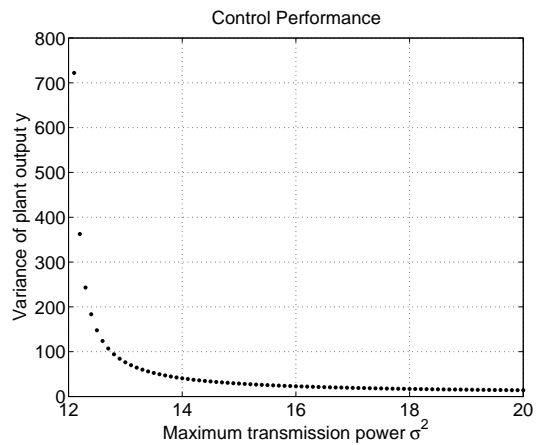}}{\includegraphics[width=0.8\hsize]{cost_by_sigma2}}
      \caption{Minimum variance of the plant output $y$ as a function of the
        SNR (or maximum allowed transmission power) $\sigma^2$, for the
        plant \mbox{$G=1/(z(z-2))$}. The variance grows unbounded as $\sigma^2$
        approaches the lower limit for stabilization.}
      \label{fig:performance}
    \end{center}
  \end{figure}

\section{Conclusion}\label{sec:conclusion}
This paper has considered a special class of decentralized control problems where the controller is split in two parts that are separated by a noisy communication channel with an SNR constraint. It has been shown that an optimal linear design can be obtained with arbitrary accuracy by solving a convex optimization problem and performing a spectral factorization.

The results in this paper can be viewed as a generalization of some results pertaining to a communication problem that can be obtained by considering the open-loop version of the control problem. 

As mentioned in Section \ref{sec:introduction}, the problem in this paper has previously been considered in the case with an AWGN channel with feedback \cite{silva2010c}, where a similar result is obtained using a slightly different technique. A disadvantage with that result is, though, that it requires the controller to be over-parametrized with four degrees of freedom. The technique used in this paper has been applied to that case as well in \cite{joh11phd}, giving a solution that does not require an over-parametrization.

Objects for further research include an extension to handle MIMO channels or plants with more than one controller input or measurement signal. Of course it would also be of interest to know if non-LTI controllers could provide better performance or require lower SNR levels for stabilization when the channel noise is colored.

\appendix
\begin{lem}\label{lem:product_with_inverse_in_Hp}
  Suppose that $Y \in \Smirnov$ is square and outer, $X \in \Smirnov$, and that \mbox{$Y^{-1}X \in \Lebesgue{p}$}. Then $Y^{-1}X \in \Hardy{p}$.
\end{lem}
\begin{proof}
  $Y^{-1} \in \Smirnov$ by Theorem 10 in \cite{Inouye}. It is easy to verify that the product of two $\Smirnov$ functions is $\Smirnov$. The result follows from $\Lebesgue{p} \cap \Smirnov = \Hardy{p}$ \cite{Garnett}.
\end{proof}

\begin{proof}[Proof of Lemma \ref{lemfnfQhat}]
  The proof is based on construction of $\hat{Q}$ through a perturbation of $Q$. Take $Q \in \Theta_Q$ and let 
  \begin{equation*}
    K=\frac{MQ-U}{NQ+V}.
  \end{equation*}
  If $K\in\RLone$ then let $\hat{Q}=Q$ and the construction is complete. 
  Suppose instead that $K$ has at least one pole on $\T$. Since $MQ-U\in\RHinf$, $z$ is a pole of $K$ if and only if 
  \begin{equation}\label{eq:control_Qdenom_zero}
    N(z)Q(z)+V(z)=0.
  \end{equation}
  Moreover, suppose that (\ref{eq:control_Qdenom_zero}) holds and that $N(z) = 0$. Then it follows from the Bezout identity that $V(z) \neq 0$, which is a contradiction. Thus if $NQ+V$ has a zero at $z$ then $N(z) \neq 0$. 

  Suppose now that $NQ+V$ has a zero at $z_0\in\mathbb{T}$ and that $z_0\notin \R$ (the case when $z_0\in\R$ is discussed later). Let 
  \begin{equation*}
    \hat{Q}=Q+\lambda_0+\lambda_1z^{-1}, \quad \lambda_0,\lambda_1 \in \R.
  \end{equation*}
  Then $\norm{E\hat{Q}+F}_2 < \sigma^2+1$ if $|\lambda_0|+|\lambda_1|<\delta_\lambda$ for small enough $\delta_\lambda$. 

  The coefficients $\lambda_0,\lambda_1$ will be chosen so that the zero at $z_0$ is perturbed away from $\T$. It must also be made sure that none of the other zeros can reach $\T$ under the same perturbation. For this reason, define the set of zeros not on the unit circle, 
  \begin{equation*}
    \Omega = \lbrace z:z\notin \T, N(z)Q(z)+V(z)=0 \rbrace,
  \end{equation*}
  and the smallest distance from that set to the unit circle,
  \begin{equation*}
    r = \inf_{z_1 \in \Omega, z_2 \in \T} \left | z_1-z_2 \right |,
  \end{equation*}
  where $r>0$ since $\Omega$ has a finite number of elements. The location of the zeros of $N\hat{Q}+V$ depend continuously on $(\lambda_0,\lambda_1)$. Thus, there exists $\delta_r>0$ such that if $|\lambda_0|+|\lambda_1|<\delta_r$ then all zeros are displaced strictly less than $r$.

  Introduce the function
  \begin{equation*}
    X(z,\lambda_0,\lambda_1) = N\hat{Q}+V = NQ+V+N(\lambda_0+\lambda_1z^{-1}).
  \end{equation*}
  Then
  \begin{equation*}
    \det \begin{bmatrix}
      \re \dfrac{\partial X}{\partial \lambda_0} & \re \dfrac{\partial X}{\partial \lambda_1} \vspace{1mm} 
      \\
      \im \dfrac{\partial X}{\partial \lambda_0} & \im \dfrac{\partial X}{\partial \lambda_1} 
    \end{bmatrix} = \det
    \begin{bmatrix}
      \re N & \re Nz^{-1} \\
      \im N & \im Nz^{-1}
    \end{bmatrix} 
  \end{equation*} is non-zero at $z=z_0$
  since $N(z_0)\neq 0$ and \mbox{$z_0\in\T\setminus\R$}.
  Then, by the implicit function theorem, there is a differentiable mapping $z \mapsto (\lambda_0,\lambda_1)$ defined in a neighborhood of $z_0$, such that
  \begin{multline*}
    N(z)\hat{Q}(z)+V(z) \\ = 
    N(z)Q(z)+V(z)+N(z)(\lambda_0(z)+\lambda_1(z)z^{-1}) = 0.
  \end{multline*}
  This means that a new location $z$ can be determined for the zero, and the mapping gives the corresponding $\lambda_0,\lambda_1$.

  Take $\varepsilon>0$. Since $\varphi(Q)$ is continuous
  there exists $\delta_Q>0$ such that
  \begin{equation*}
    \norm{\hat{Q}-Q}_\infty < \delta_Q \Rightarrow \left|\varphi(\hat{Q}) - \varphi(Q) \right | < \varepsilon.
  \end{equation*}
  Continuity of the mapping from $z$ to $(\lambda_0,\lambda_1)$ implies that there exists $\delta_z>0$ such that 
  \begin{equation*}
    \left |z-z_0\right |<\delta_z \Rightarrow \left | \lambda_0(z) \right | + \left | \lambda_1(z) \right | < \min\lbrace\delta_Q,\delta_\lambda,\delta_r\rbrace.
  \end{equation*}
  Now pick $z\notin\T$ such that $|z-z_0 |<\delta_z$ and the mapping to $\lambda_0, \lambda_1$ is defined. Then  
  \begin{equation*}
    \norm{\hat{Q}-Q}_\infty \leq \left | \lambda_0(z) \right | + \left | \lambda_1(z) \right | < \min\lbrace\delta_Q,\delta_\lambda,\delta_r\rbrace,
  \end{equation*}
  which implies that \[\norm{E\hat{Q}+F}_2 < \sigma^2+1, \quad |\varphi(\hat{Q}) - \varphi(Q) | < \varepsilon,\] and that there are no new zeros on $\T$. Since $z\notin \T$ it follows that $N\hat{Q}+V$ has at least one zero less than $NQ+V$ on $\T$. 

  If $z_0$ is real, then define instead
  \begin{equation*}
    \hat{Q} = Q+\lambda_0, \quad \lambda_0 \in \R
  \end{equation*}
  and determine $\lambda_0$ analogously. Note, however, that the zero must be kept on the real axis.

  If $\hat{Q}$ is such that $N\hat{Q}+V$ has zeros on $\T$, the procedure may be repeated again, with $\varepsilon$ appropriately chosen, until there are no such zeros. Thus, for every $Q \in\Theta_Q$ and $\varepsilon>0$ it is possible to construct $\hat{Q}$ such that $N\hat{Q}+V$ has no zeros on $\T$, $|\varphi(\hat{Q}) - \varphi(Q) | < \varepsilon$ and $\norm{E\hat{Q}+F}_2 < \sigma^2+1$. 
\end{proof}

\section*{Acknowledgment}
The authors would like to thank Eduardo Silva (UTFSM) and Alexandre Megretski (MIT) for helpful comments and technical discussions.

\ifCLASSOPTIONcaptionsoff
  \newpage
\fi


\end{document}